\documentclass[final]{siamltex}
\usepackage[parfill]{parskip}    
\usepackage{amsmath,amssymb,latexsym,enumerate,mathrsfs}
\usepackage{hyperref}
\usepackage{array}
\usepackage{times}
\usepackage[]{mdframed}
\usepackage{framed}
\usepackage{subfig}
\usepackage{graphicx}
\usepackage{hyperref} 
\usepackage{multicol}
\usepackage{framed}
\usepackage{moreverb}
\usepackage{marvosym}
\usepackage{color,soul}
\definecolor{lightblue}{rgb}{.90,.95,1}
\sethlcolor{yellow}

\usepackage{epstopdf}
\DeclareMathOperator*{\argminmax}{arg\,min\,max}

\newtheorem{assumption}{\bf Assumption}

\title{A perspective note on likelihood approximation and inference for complex simulation models using a chain of aggregated normalizing flows}

\author{Getachew K. Befekadu\footnote{\scriptsize Getachew K. Befekadu is with the Department of Electrical \& Computer Engineering, College of Engineering, Physics, and Computing, The Catholic University of America, ~ Washington, DC 20064, USA. E-mail:\,{\tt befekadu@cua.edu}}}

\begin{document}
\maketitle

\renewcommand{\thefootnote}{\arabic{footnote}}

\begin{abstract}
In this work, we present a new perspective on the problem of likelihood approximation within the framework of simulation-based inference that promotes scalable and controllable simulation routines for large-scale data analysis, allows efficient parameter space exploration or smooth interpolation in high-dimensions and, thus, supports valid statistical treatments of hypothesis testings as well as uncertainty quantification. In particular, we consider a chain of $n$-aggregated normalizing flows (i.e., a set of well-structured bijective transformations admitting a sequence of functional composition and aggregating estimated parameters) for likelihood approximation scheme, where a set of upfront replicated observation datasets from the forward complex simulation model pass through the first set of bijective transformations, and then subsequently pass to the other sets of bijective transformations. Here, we assume that, for any $k \in \{1,\,2, \ldots, n\}$, the parameters corresponding to the first $k$ sets of bijective transformations are estimated sequentially, in some sense of optimality, for constructing flexible probability distributions, regardless of the remaining $(n-k)$ sets of bijective transformations. Moreover, our objects of interest are to highlight two complementary mathematical arguments that leverage an informatics-theoretic formalization, based-on empirical likelihood estimators under moment restrictions, and a sequential decision-making paradigm, with a set of mixing distribution strategies, for updating and aggregating the estimated parameters of the overall normalizing flows. Finally, as a by-product, the framework presented in this work also provides a reliable surrogate model, conditioned on the model parameters defining the forward computational simulation, that allows samples generation, with statistical powers, and further facilitates computationally tractable scheme in the Bayesian paradigm for inference, hypothesis testings and uncertainty quantification.
  
\end{abstract}

\section{Introduction} \label{S1}
With the advent of high performance computing and highly expressive programming languages, researchers as well as practitioners from various domains of applied science and engineering are increasingly relying on complex simulation models as an integral part of their broader scientific investigations that include interpreting observation datasets, forecasting system behaviors, quantifying uncertainties, synthesizing plausible hypotheses or making decisions, and optimal experiment designs. Although these complex, and often large-scale, simulation models provide useful information about the underlying real physical systems or experimental processes they represent. It is also acknowledged that, due to the lack of closed form or intractability of the corresponding exact likelihood functions, they are poorly suitable by themselves for statistical inference tasks -- such as assessing the predictive capabilities for exploratory scenarios or estimating model parameters for the underlying inverse problems -- that are required evidently for a rigorous uncertainty quantification on quantities of interest (e.g., see \cite{r1}-\cite{r7} for interesting studies discussing the need of addressing such challenges). As a result of this, developing a viable statistical inference framework in the Bayesian paradigm, that provides reliable solutions for quantities of interest and their uncertainties, while computationally tractable, i.e., admitting computationally efficient and scalable routines for constructing flexible and/or locally approximating target probability distributions, is certainly indispensable for both researchers as well as practitioners working in areas of model verification, validation and uncertainty quantification. 

In this work, without claiming as a one-size-fits-all method for statistical inference, we provide a new perspective on the problem of likelihood approximation within the framework of simulation-based inference that promotes scalable and controllable simulation routines for large-scale data analysis, allows efficient parameter space exploration or smooth interpolation in high-dimensions and, thus, supports valid statistical treatments of hypothesis testings as well as uncertainty quantification. In particular, we consider a chain of $n$-aggregated normalizing flows, i.e., a set of well-structured bijective transformations admitting a sequence of functional composition and aggregating estimated parameters, for likelihood approximation scheme, where a set of upfront replicated observation datasets from the forward complex simulation model pass through the first set of bijective transformations, and then subsequently pass to the other sets of bijective transformations. Moreover, we assume that, for any $k \in \{1,\,2, \ldots, n\}$, the parameters corresponding to the first $k$ sets of bijective transformations are estimated sequentially, in some sense of optimality, for constructing flexible probability distributions, regardless of the remaining $(n-k)$ sets of bijective transformations. 

Moreover, our objects of interest include highlighting the rationale behind two complementary mathematical arguments that leverage an informatics-theoretic formalization, based-on empirical likelihood estimators under moment restrictions, and a sequential decision-making paradigm, with a set of mixing distribution strategies, for updating and aggregating the estimated parameters of the overall normalizing flows. Here, it is worth mentioning that some interesting studies on normalizing flows as a class of generative models for continuous random variables, that provide a mechanism for representing or defining expressive probability distributions, have been reported in machine learning literature (e.g., see \cite{r8}-\cite{r14} for some interesting discussions). Note that the rationale behind our framework follows, in some sense, the settings of these papers, but in a different context. Finally, as a by-product, the framework presented in this work also provides a reliable surrogate model, conditioned on the model parameters defining the forward computational simulation, that allows samples generation, with statistical powers, and further facilitates computationally efficient and tractable scheme in the Bayesian paradigm for inference, hypothesis testings and uncertainty quantification.

\subsection*{Outline} In Section~\ref{S2}, we present background information that are relevant for the exposition of a new perspective on the problem of likelihood approximation and inference within the framework of simulation-based inference. In particular, this section outlines some of the core concepts concerning a chain of $n$-aggregated normalizing flows for likelihood approximation scheme, with the focus of upfront replicated observation datasets from the forward computational simulation, with known parameters sampled from a prior distribution. Here, it is also assumed that the replicated observation datasets from the forward computational simulation pass through the first $k$ set of bijective transformations, and then subsequently pass to the other $(n-k)$ sets of bijective transformations, with $k \in \{1,\, 2, \ldots, n\}$. In Section~\ref{S3}, we present our main results along with our mathematical arguments that leverage an informatics-theoretic formalization and a sequential decision-making paradigm with mixing distributions for updating and aggregating the estimated parameters of the overall normalizing flows. While, such overall normalizing flows support smooth interpolation of parameters in high dimensions, they also provide an extension mechanism for constructing flexible and/or locally approximating target probability distributions that serve as a proxy to the underlying true, but intractable, likelihood functions. Moreover, this section also provides a reliable surrogate model, conditioned on the model parameters defining the forward computational simulation, that allows samples generation, with statistical powers, and further facilitates computationally tractable scheme in the Bayesian paradigm for statistical inference, hypothesis testings and uncertainty quantification. Finally, this section provides comments on model inference tasks, in the parameters setting, as well as quantifying model performances. 

\section{Background and preliminaries} \label{S2}
In this section, we provide background information that are relevant for approximating the likelihood function (that serves as a proxy to the underlying true likelihood function) for the complex simulation model within the framework of simulation-based inference. In particular, the framework considered in this work, which leverages the use of an informatics-theoretic formalization and a sequential decision paradigm with mixing distributions, consists of the following core concepts:
\begin{enumerate} [(i).]
\item {\it Replicated simulation observation datasets as training data}: We assume that we have upfront replicated observation datasets, i.e., $R$ replicated observation datasets of $y_{\theta}^{r}$, for $r=1$, $2$, \ldots, $R$, with the support of $Y \subset \mathbb{R}^{d_y}$, from the forward computational simulation corresponding to a certain physical system or real experimental process; while $\theta$ denotes parameters in the model defining the forward computational simulation. In particular, these replicated observation datasets, each with a data size of $N$, generated by the forward computational simulation, consist of a set of observation datasets $y_{\theta^{(r)}}^{(r)} = \bigl\{y_{\theta^{(r)}, j}^{(r)} \bigr\}_{j=1}^N$ at different sampled parameters values of $\theta^{(r)} \sim \pi(\theta)$, for $r=1$, $2$, \ldots, $R$, from a prior distribution $\pi(\theta)$ on the parameters, with the support of $\Theta \subset \mathbb{R}^{d_{\theta}}$.
\item {\it A chain of aggregated normalizing flows for likelihood approximation scheme}: We consider a chain of $n$ aggregated normalizing flows (i.e., a set of well-structured bijective transformations admitting a sequence of functional composition and aggregating estimated parameters) for likelihood approximation scheme, where the replicated observation datasets from the complex simulation model pass through the first set of bijective transformations $\operatorname{T}_{\phi^{(r, k)}} (\cdot)$, for $r =1$, $2$, \ldots, $R$, with $k=1$, and then subsequently pass to the other sets of bijective transformations $\operatorname{T}_{\phi^{(r, k)}} (\cdot)$, for $r =1$, $2$, \ldots, $R$, with $k \in \{2, \ldots, n\}$. Here, we assume that, for any $k \in \{1,\,2, \ldots, n\}$, the parameters $\phi^{(r, k)} = \bigl\{\alpha^{(r, k)}, \gamma^{(r, k)}\bigr\}$, with the scaling parameters $\alpha^{(r, k)} > 0$ and shifting parameters $\gamma^{(r,k)} \in \mathbb{R}$, for $r=1$, $2$, \ldots, $R$, that correspond to the first $k$ set of bijective transformations are determined, in some sense of optimality to be discussed later, regardless of the remaining $(n-k)$ set of bijective transformations, for constructing flexible probability distributions. Furthermore, by dropping the dependence on the model parameters $\theta$ for sake of clarity, the transformations of the replicated observation datasets will assume the following forms
\begin{align}
z^{(r, k)} &= \operatorname{T}_{\phi^{(r, k)}} \bigl(z^{(r, k-1)}\bigr) \notag\\
              &= \bigl\{z_j^{(r, k-1)} \exp\bigl(\alpha^{(r, k)}\bigr) + \gamma^{(r, k)}\bigr\}_{j=1}^N,  \quad r=1, 2, \ldots, R, \label{Eq2.1}
\end{align}
with $k \in \{1,\, 2, \ldots, n\}$, where $z^{(r, 0)}$ are given by $z^{(r, 0)}=y^{(r)}$, for $r =1$, $2$, \ldots, $R$. Moreover, the parameters aggregation step, that leverages a sequential decision-making paradigm, with a set of mixing distribution strategies, will then proceed with respect to the current estimated parameters $\phi^{(r, k)} = \bigl\{\alpha^{(r, k)}, \gamma^{(r, k)}\bigr\}$, for $r=1$, $2$, \ldots, $R$, corresponding to the $k^{\rm th}$ set of bijective transformations. Note that, by treating the upfront replicated observation datasets, i.e., $\bigl\{\theta^{(r)}, y^{(r)}\bigr\}_{r=1}^R$, as training data, we view the aggregated overall normalizing flows as an instrument for approximating the true likelihood function corresponding to the complex simulation model. 
\item {\it Empirical-likelihood estimators under moment restrictions for parameters estimation scheme}: We consider a family of empirical-likelihood estimators, under moment restrictions, corresponding to the block-wise implementations of the transformed replicated observation datasets of
\begin{align}
z^{(r, 0)}=y^{(r)}, \quad z^{(r, 1)}=\operatorname{T}_{\phi^{(r, 1)}} \bigl(z^{(r,0)}\bigr), \quad \cdots, \quad z^{(r, n)}=\operatorname{T}_{\phi^{(r, n)}} \bigl(z^{(r, n-1)}\bigr), \label{Eq2.2}
\end{align}
where 
\begin{align}
z^{(r, k)} &= \operatorname{T}_{\phi^{(r, k)}} \bigl(z^{(r, k-1)}\bigr) \notag \\
              &= \bigl\{z_j^{(r, k-1)} \exp\bigl(\alpha^{(r, k)}\bigr) + \gamma^{(r, k)}\bigr\}_{j=1}^N, \quad r=1, 2, \ldots, R,  \label{Eq2.3}
\end{align}
with $k \in \{1,\, 2, \ldots, n\}$. That is, for an integer $m$ (with $m < N$), the block-implementations of the transformed replicated observation datasets $\hat{z}_{\ell}^{(r, k)}$ are computed as follows
\begin{align}
 \hat{z}_{\ell}^{(r, k)} = \bigl(z_{\ell}^{(r, k)}, \, z_{\ell+1}^{(r, k)}, \ldots, z_{\ell+m-1}^{(r, k)} \bigr), &\quad \ell=1, \ldots, (N-m+1), \label{Eq2.4}\\
                                                                                                                                                       &\quad r=1, 2, \ldots, R, \notag
\end{align}
with $k \in \{1,\, 2, \ldots, n\}$. Moreover, each of the block-implementations are assumed to provide a basis for retaining the dependence patterns of $z_{\ell}^{(r, k)}$ in each of the blocks $\hat{z}_{\ell}^{(r, k)}$ of length $m$, for $\ell=1$, \ldots, $(N-m+1)$ and further allowing us to compute the corresponding smoothed aggregated (moment) functions, that are appropriate for the empirical-likelihood estimators, under moment restrictions. Note that the reasoning behind this instantiated step is to impose indirect constraints on the distributions $p_{\rm z^{(k,r)}}(z)$ corresponding to the transformed observation datasets $z^{(r, k)}$, i.e.,
\begin{align}
p_{z^{(r, k-1)}} \bigl(z \vert \phi^{(r, k-1)} \bigr) &= p_{z^{(r, k)}} \bigl(\operatorname{T}_{\phi^{(r, k-1)}} (z) \vert \phi^{(r, k)} \bigr) \biggl \vert \operatorname{det} \frac{\partial \operatorname{T}_{\phi^{(r, k)}}(z)} {\partial z} \biggr \vert,& \label{Eq2.5}\\
                      \quad r=1, 2, \ldots, R,& \notag
\end{align}
with $k \in \{1, 2, \ldots, n\}$, where $\operatorname{det} {\partial \operatorname{T}_{\phi^{(r, k)}} (z)}/{\partial z}$, for $r=1$, $2$, \ldots, $R$, with $k \in \{1, 2, \ldots, n\}$, are the determinant for the $d_{y} \times d_{y}$ Jacobian matrices. Furthermore, at the same time, such an instantiated step reasonably allows us to treat the transformed block-wise observation datasets $\hat{z}_{\ell}^{(r, k)}$, for $\ell=1$, \ldots, $(N-m+1)$ and $r=1$, $2$, \ldots, $R$, as random variables arising from discretely distribution populations with each transformed block-wise observation datapoint as a single-cell from a set of finite-cell contingencies. Later, based on this argument, we further establish a connection between the dual solvability conditions of a family optimization problems corresponding to the empirical-likelihood estimators, under moment restrictions, with that of saddle-point optimality conditions for a certain family of {\it minmax} optimization problems. Moreover, such an instantiated step will also allow us to collect all gradient-related information for updating the estimated parameters $\bigl\{\phi^{(r, k)}\bigr\}_{r=1}^R = \bigl\{\alpha^{(r, k)}, \gamma^{(r, k)}\bigr\}_{r=1}^R$, with $k \in \{1,\,2, \ldots, n\}$, of the normalizing flows using first-order gradient methods that are preferable for large-scale applications. Here, the updating schemes, for each iteration time step $t=1$, $2$, \ldots, $T-1$, will follow
\begin{align}
\phi_{t+1}^{(r, k)} = \phi_{t}^{(r, k)} - \delta_t^{(r, k)} \left(\nabla_{\phi} \mathcal{L}^{(r, k)} \bigl(\phi^{(r, k)}\bigr) + \nu_{\phi} \nabla_{\phi} \mathcal{R}_{\phi} \bigl(\phi^{(r, k)} \bigr) \right)\biggl \vert_{\phi_{t}^{(r, k)}}, \label{Eq2.6}\\
                      \quad r=1, 2, \ldots, R,& \notag
\end{align}
with $k \in \{1, 2, \ldots, n\}$, where $\mathcal{L}^{(r, k)}$, for $r=1$, $2$, \ldots, $R$, with $k \in \{1, 2, \ldots, n\}$, are the empirical log-likelihood functions associated with the implied empirical probabilities from the empirical likelihood estimators, under moment restrictions; while $\mathcal{R}_{\phi}$ is the regularization term.
\item{\it A sequential decision-making paradigm with mixing distributions for parameters aggregation scheme}: We consider an aggregation scheme for the current estimated parameters of the bijective transforms. Here, at each iteration time step $t = 0, 1, 2, \ldots, T-1$, we determine a set of mixing distribution strategies
\begin{align}
 \omega_t^{(k)}=\bigl\{\omega_t^{(1, k)}, \omega_t^{(2, k)}, \ldots, \omega_t^{(R, k)}\bigr\}, \label{Eq2.7}
\end{align}
where $\omega_t^{(k, r)} > 0$, for $r=1$, $2$, \ldots, $R$, with $k \in \{1,\,2, \ldots, n\}$, sum to one, i.e., $\sum\nolimits_{r=1}^R \omega_t^{(r, k)} = 1$, that will be used for dynamically apportioning weighting coefficients with respect to the current estimated parameters $\bigl\{\phi_t^{(r, k)}\bigr\}_{r=1}^R = \bigl\{\alpha_t^{(r, k)}, \gamma_t^{(r, k)}\bigr\}_{r=1}^R$ for the $k^{\rm th}$ set of bijective transforms. That is, the aggregation scheme for the current estimated parameters proceeds as
\begin{align}
\bar{\phi}_t^{(r, k)} & = \left\{ \bar{\alpha}_t^{(k)}, \bar{\gamma}_t^{(k)} \right\} \notag\\
                           & = \left\{ \sum\nolimits_{r=1}^R \omega_t^{(r, k)} \alpha_{t}^{(r, k)}, \sum\nolimits_{r=1}^R \omega_t^{(r, k)}\gamma_{t}^{(r, k)}  \right\},  \label{Eq2.8}
\end{align}
for each $k \in \{1, 2, \ldots, n\}$. 

Moreover, our object of interest here includes providing conditions under which such an aggregation scheme will sequentially determine a set of mixing distribution strategies used for aggregating the current estimated parameters for the overall normalizing flows that ultimately leading to optimal parameters estimates, i.e., further allowing smooth interpolation of parameters in high-dimensions and supporting an extension mechanism for constructing flexible probability distributions.
 \end{enumerate}

In the following section (i.e., Section~\ref{S3}), we emphasize that estimating the parameters for the aggregated overall normalizing transforms, i.e.,
\begin{align*}
\phi^{(r, k)} = \bigl\{ \alpha^{(r, k)}, \gamma^{(r, k)} \bigr\} \quad r=1, 2, \ldots, R,
\end{align*}
with $k \in \{1, 2, \ldots, n\}$, is linked with the solutions of the empirical-likelihood estimators, under moment restrictions, where such solutions are also assume to provide indirect constraints on the probability distributions $p_{z^{(r, n)}} \bigl(z \vert \phi^{(r, k)} \bigr)$, for $r=1$, $2$, \ldots, $R$, with $k \in \{2, \ldots, n\}$, corresponding to the transformed replicated observation datasets $z^{(r, k)} = \bigl\{z_j^{(r, k-1)} \exp\bigl(\alpha^{(r, k)}\bigr) + \gamma^{(r, k)}\bigr\}_{j=1}^N$, for $r =1$, $2$, \ldots, $R$, with $k \in \{2, \ldots, n\}$.

\section{Main results} \label{S3}
In this section, we present our main results concerning a chain of aggregated normalizing flows that provides a mechanism for approximating the likelihood function (i.e., serving as a proxy to the underlying true likelihood function) within the framework of simulation-based inference. Moreover, Figure~\ref{Figure:FG1} shows a conceptual workflow that can be implemented with any highly expressive programming language.\footnote{In this work, our intent is to provide the rationale behind our mathematical arguments, rather than considering a specific simulation problem or application.}

\begin{figure}[hbt]
\begin{center}
 \includegraphics[scale=0.75]{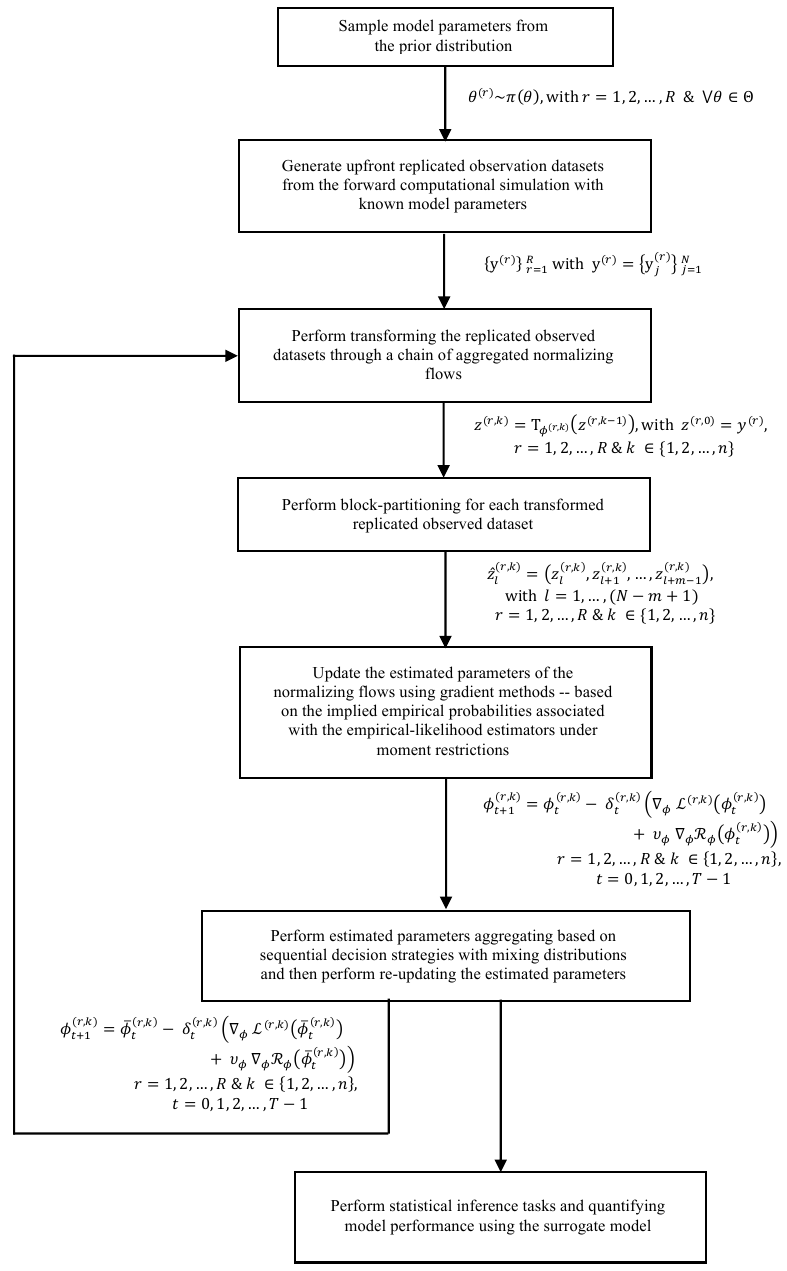}
 \caption{\small A conceptual workflow for approximation likelihood function and inference tasks within the framework of simulation-based inference.} \label{Figure:FG1}
\end{center}
\end{figure}

\subsection{A chain of aggregated normalizing flows for likelihood approximation scheme} \label{S3.1}
In this subsection, we restrict our discussion with regard to a chain of aggregated normalizing flows, i.e., a mapping mechanism linking the replicated observation datasets with that of the model parameters for the complex simulation model. In particular, for $R$ upfront replicated observed datasets from the forward computational simulation for known parameters, we assume that such replicated observation datasets pass through the first set of bijective transformations $\operatorname{T}_{\phi^{(r, k)}} (\cdot)$, for $r =1$, $2$, \ldots, $R$, with $k=1$, and then subsequently pass to the other sets of bijective transformations $\operatorname{T}_{\phi^{(r, k)}} (\cdot)$, for $r =1$, $2$, \ldots, $R$, with $k \in \{2, \ldots, n\}$. Note that, by treating the upfront replicated observation datasets as training data, we view the aggregated overall normalizing flows as an instrument for approximating the likelihood function corresponding to the complex simulation model. Moreover, we assume that,  for $r=1$, $2$, \ldots, $R$, with $k \in \{1,\,2, \ldots, n\}$, the parameters $\phi^{(r, k)} = \bigl\{\alpha^{(r, k)}, \gamma^{(r, k)}\bigr\}$, with the scaling parameters $\alpha^{(r, k)} > 0$ and shifting parameters $\gamma^{(r,k)} \in \mathbb{R}$, that correspond to the first $k$ set of bijective transformations are determined, in some sense of optimality, regardless of the remaining $(n-k)$ set of bijective transformations for constructing flexible probability distributions. Note that, by dropping the dependence on the model parameters $\theta$, the transformations of the replicated observation datasets will assume the following forms
\begin{align}
z^{(r, k)} &= \operatorname{T}_{\phi^{(r, k)}} \bigl(z^{(r, k-1)}\bigr), \notag \\
              &= \bigl\{z_j^{(r, k-1)} \exp\bigl(\alpha^{(r, k)}\bigr) + \gamma^{(r, k)}\bigr\}_{j=1}^N, \quad r=1, 2, \ldots, R, \label{Eq3.1}
\end{align}
with $k \in \{1,\, 2, \ldots, n\}$; while $z^{(r, 0)}$ are given by $z^{(r, 0)}=y^{(r)}$, for $r =1$, $2$, \ldots, $R$. Moreover, the parameters aggregation step, which is based on a sequential decision-making paradigm with a set of mixing distribution strategies, will then proceed with respect to the current estimated parameters $\phi^{(r, k)} = \bigl\{\alpha^{(r, k)}, \gamma^{(r, k)}\bigr\}$, for $r=1$, $2$, \ldots, $R$, corresponding to the $k^{\rm th}$ set of bijective transformations (cf. Subsection~\ref{S3.3}).

Here, we assume that the following assumption holds true.
\begin{assumption} \label{AS1}
For known parameters $\theta$, sampled from the prior distribution $\pi(\theta)$, with the support of $\Theta \subset \mathbb{R}^{d_{\theta}}$, the bijective transformations $\operatorname{T}_{\phi^{(r, k)}}$, parametrized with respect to $\phi^{(r, k)}$ satisfy the following conditions
\begin{align}
p_{z^{(r, k-1)}} \bigl(z \vert \phi^{(r, k-1)} \bigr) = p_{z^{(r, k)}} \bigl(\operatorname{T}_{\phi^{(r, k-1)}} (z) \vert \phi^{(r, k)} \bigr) \biggl \vert \operatorname{det} \frac{\partial \operatorname{T}_{\phi^{(r, k)}}(z)} {\partial z} \biggr \vert,& \label{Eq3.2}\\
                                                r=1, 2, \ldots, R,& \notag
\end{align}
with $k \in \{1, 2, \ldots, n\}$, where $p_{z^{(r, 0)}} \bigl(z \vert \phi^{(r, 0)} \bigr)$, for $r =1$, $2$, \ldots, $R$, represent the unknown complex probability distributions corresponding to the original replicated observation datasets $z^{(r, 0)} = y^{(r)}$, for $r =1$, $2$, \ldots, $R$ and $\operatorname{det} {\partial \operatorname{T}_{\phi^{(r, k)}} (z)}/{\partial z}$ denotes the determinant for the Jacobian matrices. Moreover, the probability distributions $p_{z^{(r, k-1)}}$, for $r =1$, $2$, \ldots, $R$, with $k \in \{2, \ldots, n\}$, correspond to the transformed replicated observation datasets $z^{(r, k-1)}$, whose properties constrained by a sequence of solutions from the family of the empirical-likelihood estimators, under moment restrictions.
\end{assumption}

In the following subsections (i.e., Subsections~\ref{S3.2} and \ref{S3.3}), we present our mathematical arguments for parameters updating and aggregating schemes, that leverage an informatics-theoretic formalization and a sequential decision-making paradigm leading to the optimal estimated parameters $\phi^{\ast} = \bigl\{ \bar{\alpha}_{T}^{(k)},\bar{\gamma}_{T}^{(k)} \bigr\}_{k=1}^n$ for the aggregated overall normalizing flows. Note that a typical estimation scheme for obtaining $\phi^{\ast}$ involves minimizing the Kullback-Leibler (KL) divergence between the true likelihood function $p\bigl(y\vert \theta\bigr)$ and the likelihood approximation $\mathcal{L}_{\rm Approx} \bigl(y_{\rm obs} \vert \theta, \phi^{\ast} \bigr)$ associated with the aggregated overall normalizing flows, conditioned on the model parameters $\theta$ defining the forward computational simulation ( (e.g., see \cite{r15}-\cite{r17} for general discussions on KL divergence). Here, instead of relying on minimizing the KL divergence, we provide parameters updating and aggregating schemes that utilize $R$ upfront replicated observation datasets $\bigl\{\theta^{(r)}, y^{(r)}\bigr\}_{r=1}^R$, with $\theta^{(r)} \sim \pi(\theta)$, for $r=1$, $2$, \ldots, $R$, as training data and further allow us to view the aggregated overall normalizing flows as an instrument for approximating the likelihood function associated with the complex simulation model. 

\subsection{Empirical-likelihood estimators under moment restrictions for parameters estimation scheme} \label{S3.2}
In this subsection, with respect to the block-implementations of the transformed replicated observation datasets $\bigl\{\hat{z}_{\ell}^{(r, k)}\bigr\}_{\ell=1}^{N-m+1}$, for $r=1$, $2$, \ldots, $R$, with $k \in \{1,\, 2, \ldots, n\}$, we consider a family of optimization problems (i.e., the empirical likelihood estimators, under moment restrictions) of the following form
\begin{align}
\max_{\pi_\ell^{(r, k)}, 1\le \ell \le N-m+1; \, \beta^{(r, k)}} &\,\sum\nolimits_{\ell=1}^{N-m+1} \log \pi_\ell^{(r, k)} \label{Eq3.3}\\
&\text{s.t.} \notag\\
& \quad \sum\nolimits_{\ell=1}^{N-m+1} \pi_\ell^{(r, k)} \psi\bigl(\hat{z}_\ell^{(r, k)}, \beta^{(r, k)} \bigr) = 0, \notag\\
& \quad \sum\nolimits_{\ell=1}^{N-m+1} \pi_\ell^{(r, k)} = 1, \notag
\end{align}
where $\pi_\ell^{(r, k)} > 0$, for $\ell=1$, \ldots, $(N-m+1)$, are the implied empirical probabilities satisfying $\sum\nolimits_{\ell=1}^{N-m+1} \pi_\ell^{(r, k)} = 1$, for $r=1$, $2$, \ldots, $R$, with $k \in \{1,\, 2, \ldots, n\}$.Moreover, $\psi(\cdot, \cdot)$ denotes smoothed aggregated (moment) functions and defined as follows
\begin{align}
 \psi \bigl(\hat{z}_\ell^{(r, k)}, \beta^{(k,r)} \bigr) = \frac{1}{m} \sum\nolimits_{s=0}^{m-1} g\bigl(z_{\ell+s}^{(r, k)}, \beta^{(r, k)} \bigr), \quad \ell = 1, \ldots, (N-m+1),& \label{Eq3.4}\\
  r=1, 2, \ldots, R,& \notag                     
\end{align}
with $k \in \{1,\, 2, \ldots, n\}$, where $\beta^{(r, k)} \in \mathbf{B}$, with compact set $\mathbf{B} \subset \mathbb{R}^{d_{\beta}}$, for $r=1$, $2$, \ldots, $R$, with $k \in \{1,\, 2, \ldots, n\}$, and $g \colon \mathbb{R}^{d_y} \times \mathbf{B} \to \mathbb{R}^{d_g}$ is a measurable function and typically interpreted as a vector of residuals for modeling of moment restrictions (e.g., see \cite{r18}-\cite{r21} for additional discussions).

In what follows, we assume the following statements that are necessary for the solvability of the family of optimization problems in Equation~\eqref{Eq3.3} (e.g., see \cite{r20} for additional discussions).
\begin{assumption}\label{AS2}
\begin{enumerate} [(1).]
\item There exists a unique solution $\beta_{0} \in \operatorname{Int}(\mathbf{B})$ satisfying $E\bigl(\psi\bigl(z,\beta_{0}\bigr)\bigr) = 0 $, for $z$ with the support of $\mathbb{R}^{d_y}$.
\item The function $\psi\bigl(z,\beta \bigr)$ is continuous at each $\beta \in \mathbf{B}$ with probability one. Moreover, the following conditions hold true:
\begin{enumerate} [$~~~$(i).]
\item $E \bigl(sup_{\beta \in \mathbb{B}} \bigl\Vert \psi\bigl(z,\beta \bigr) \bigl \Vert^{\alpha} \bigr) < \infty$ for some $\alpha > 2$,
\item $E \bigl(sup_{\beta \in N_{\beta_{0}}} \bigl\Vert \nabla_{\beta} \psi\bigl(z,\beta \bigr) \bigl \Vert \bigr) < \infty$, where $N_{\beta_{0}} \subset \mathbf{B}$ denotes the neighborhood of $\beta_{0}$,
\item $\Omega = E \bigl( \psi\bigl(z,\beta\bigr) \psi\bigl(z,\beta\bigr)^T\bigr)\bigl\vert_{\beta= \beta_{0}}$ is a nonsingular matrix, and
\item$\operatorname{rank}\,E \bigl(\nabla_{\beta} \psi\bigl(z,\beta \bigr) \bigr) = d_{\beta}$,
\end{enumerate}
for $z$, with the support of $\mathbb{R}^{d_y}$.
\end{enumerate}
\end{assumption}
With a standard Lagrangian argument, the above family of optimization problems in Equation~\eqref{Eq3.3} can be further transformed to a set of non-constrained optimization problems as follows
\begin{align}
\mathcal{J}_1\bigl(\beta^{(r, k)}, \pi^{(r, k)}, \lambda^{(r, k)}, \eta^{(r, k)} \bigr) & = \sum\nolimits_{\ell=1}^{N-m+1} \log \pi_\ell^{(r, k)} \notag\\
  & - (N-m+1) (\lambda^{(r, k)})^T\sum\nolimits_{\ell=1}^{N-m+1} \pi_\ell^{(r, k)} \psi\bigl(\hat{z}_\ell^{(r, k)} \beta^{(r, k)}  \bigr) \notag\\
  & \quad \quad - \eta^{(r, k)} \left(\sum\nolimits_{\ell=1}^{N-m+1} \pi_\ell^{(r, k)} - 1 \right), \label{Eq3.5}\\
   & \quad \quad \quad \quad r=1, 2, \ldots, R, \notag  
\end{align}
with $k \in \{1,\, 2, \ldots, n\}$, where $\lambda^{(r, k)} \in \mathbb{R}^{d_g}$ and $\eta^{(r, k)} \in \mathbb{R}$ are, respectively, the Lagrange multipliers associated with the constraints in the optimization problems of Equation~\eqref{Eq3.3}. Moreover, suppose that Assumption~\ref{AS2} holds true. Then, for some fixed values of $\hat{\beta}^{(r, k)} \in \operatorname{Int}(\mathbf{B})$, for $r=1$, $2$, \ldots, $R$, with $k \in \{1, 2, \ldots, n\}$, using the envelop theorem (e.g., see \cite{r22} and \cite{r23} for additional discussions), it is easy to see that, from the KKT's optimality conditions (e.g., see also \cite{r24}), the implied empirical probabilities $\hat{\pi}_\ell^{(r, k)} > 0$, for $\ell = 1$, \ldots, $(N-m+1)$, satisfy the following primal and dual feasibility conditions
\begin{equation}
\text{Primal Feasibility:} \quad \quad \left\{ \begin{matrix} 
 \left(\sum\nolimits_{\ell=1}^{N-m+1} \hat{\pi}_\ell^{(r, k)} - 1 \right)= 0 \quad \quad \quad \quad\vspace{2mm} \\
 (N-m+1) \sum\nolimits_{\ell=1}^{N-m+1} \hat{\pi}_\ell^{(r, k)} \psi\bigl(\hat{z}_\ell^{(r, k)}, \hat{\beta}^{(r, k)} \bigr) = 0 
\end{matrix}\right. \label{Eq3.6}
\end{equation}
and 
\begin{align}
\text{Dual Feasibility:} \quad \left\{\frac{1}{\hat{\pi}_\ell^{(r, k)}} - (N-m+1) (\hat{\lambda}^{(r, k)})^T \psi\bigl(\hat{z}_\ell^{(r, k)},\hat{\beta}^{(r, k)} \bigr) - \hat{\eta}^{(r, k)} = 0 \vspace{2mm}  \right.  \notag \\
 \ell = 1,\ldots, (N-m+1), \label{Eq3.7}
\end{align}
for $r =1$, $2$, \ldots, $R$, with $k \in \{1, 2, \ldots, n\}$; while $\hat{\lambda}^{(r, k)}$ and $\hat{\eta}^{(r, k)}$ are the Lagrange multipliers associated with the constraints 
\begin{align*}
\sum\nolimits_{\ell=1}^{N-m+1} \hat{\pi}_\ell^{(r, k)} \psi \bigl(\hat{z}_\ell^{(r, k)},\hat{\beta}^{(r, k)} \bigr) = 0 \quad \text{and} \quad \sum\nolimits_{\ell=1}^{N-m+1} \hat{\pi}_\ell^{(r, k)} = 1,
\end{align*}
respectively.

Note that, if we multiply Equation~\eqref{Eq3.7} on both sides by $\hat{\pi}_\ell^{(k,r)}$ and sum over all $\ell$ (ranging from $\ell=1$ to $(N-m+1)$), i.e.,
\begin{align}
 \sum\nolimits_{j=1}^{N-m+1} \hat{\pi}_j^{(r, k)} \left(\frac{1}{\hat{\pi}_\ell^{(r, k)}} - (N-m+1) (\hat{\lambda}^{(r, k)})^T \psi\bigl(\hat{z}_\ell^{(r, k)}, \hat{\beta}^{(r, k)} \bigr)  - \hat{\eta}^{(r, k)} \right) = 0. \label{Eq3.8}
\end{align}
Then, we will have the following conditions
\begin{align}
 \hat{\eta}^{(r, k)} = - (N-m+1) \label{Eq3.9}
\end{align}
and 
\begin{align}
\hat{\pi}_\ell^{(r, k)} = \frac{1}{(N-m+1) \bigl(1 + (\hat{\lambda}^{(r, k)})^T \psi\bigl(\hat{z}_\ell^{(r, k)}, \hat{\beta}^{(r, k)} \bigr)\bigr)}, \quad \ell=1, \ldots, (N-m+1), \label{Eq3.10}
\end{align}
for all $r =1$, $2$, \ldots, $R$, with $k \in \{1,\, 2, \ldots, n\}$, that further imply the Lagrange multiplier $\hat{\lambda}^{(r, k)}$ satisfies the following condition
\begin{align}
\sum\nolimits_{\ell=1}^{N-m+1} \hat{\pi}_\ell^{(r, k)}  &= \sum\nolimits_{\ell=1}^{N-m+1}\frac{1}{(N-m+1) \bigl(1 +  (\hat{\lambda}^{(r, k)})^T \psi\bigl(\hat{z}_\ell^{(r, k)}, \hat{\beta}^{(r, k)}\bigr)\bigr)} \notag \\
                                           & = 1, \label{Eq3.11}
\end{align}
that is, confirming the sum of the implied empirical probabilities is unity. Moreover, for each $\hat{\beta}^{(r, k)} \in \operatorname{Int}(\mathbf{B})$, for $r =1$, $2$, \ldots, $R$, with $k \in \{1,\, 2, \ldots, n\}$, the Lagrange multipliers $\hat{\lambda}^{(r, k)}$ also solve the following system of equations
\begin{align}
 \sum\nolimits_{\ell=1}^{N-m+1} \frac{\psi\bigl(\hat{z}_\ell^{(r, k)}, \hat{\beta}^{(r, k)} \bigr)}{1 + (\hat{\lambda}^{(r, k)})^T \psi\bigl(\hat{z}_\ell^{(r, k)},\hat{\beta}^{(r, k)}\bigr)} = 0, \label{Eq3.12}
 \end{align}
for $r =1$, $2$, \ldots, $R$, with $k \in \{1,\, 2, \ldots, n\}$. Note that, if the first-order KKT conditions for the dual variables $\hat{\lambda}^{(r, k)}$, for $r =1$, $2$, \ldots, $R$, with $k \in \{1,\, 2, \ldots, n\}$, hold true. Then, the corresponding implied empirical probabilities will further satisfy the following sample moment restrictions 
\begin{align}
\sum\nolimits_{\ell=1}^{N-m+1} \hat{\pi}_\ell^{(r, k)} \psi\bigl(\hat{z}_\ell^{(r, k)}, \hat{\beta}^{(r, k)} \bigr) = 0. \label{Eq3.13}
\end{align}
for $r =1$, $2$, \ldots, $R$, with $k \in \{1,\, 2, \ldots, n\}$. Moreover, with a slight abuse of notation, noting that $\hat{\beta}^{(r, k)}$ and $\hat{\lambda}^{(r, k)}$ continuously depend on the parameters $\alpha^{(k, r)}$ and $\gamma^{(k, r)}$, the corresponding empirical log-likelihood functions $\mathcal{L}^{(r, k)}\bigl(\alpha^{(k, r)}, \gamma^{(r, k)}\bigr)$ can then be recovered from the implied empirical probabilities as follows
\begin{align}
\mathcal{L}^{(r, k)}\bigl(\alpha^{(r, k)}, \gamma^{(r, k)}\bigr) &= \sum\nolimits_{\ell=1}^{N-m+1} \log \hat{\pi}_\ell^{(r, k)} \notag\\
                                                        &= - \log (N-m+1) \notag\\
                                                        & \quad\quad\quad - \sum\nolimits_{\ell=1}^{N-m+1} \log \left(1 + (\hat{\lambda}^{(r, k)})^T \psi\bigl(\hat{z}_\ell^{(r, k)}, \hat{\beta}^{(r, k)} \bigr)\right), \label{Eq3.14}
\end{align}
for $r =1$, $2$, \ldots, $R$, with $k \in \{1,\, 2, \ldots, n\}$. 

Here, we further establish a connection between the dual solvability conditions for the empirical-likelihood estimators, under moment restrictions, with that of the saddle-point optimality conditions for a family of {\it minmax} optimization problems. Note that, from Assumption~\ref{AS2}, if the function $\psi(\cdot, \beta)$ is continuously differentiable with respect to $\beta \in \mathbf{B}$, then the saddle-point optimality conditions for the following family of {\it minmax} optimization problems
\begin{align}
\bigl(\hat{\beta}^{(r, k)}, \hat{\lambda}^{(r, k)} \bigr) \in \argminmax_{~~\beta^{(r, k)} \in \mathbf{B},\, \lambda^{(r, k)} \in \mathbb{R}^{d_g}} \sum\nolimits_{\ell=1}^{N-m+1} \rho \left( (\lambda^{(r, k)})^T \psi\bigl(\hat{z}_\ell^{(r, k)}, \beta^{(r, k)}\bigr)\right), \label{Eq3.15}
\end{align}
for $r =1$, $2$, \ldots, $R$, with $k \in \{1,\, 2, \ldots, n\}$, where $\rho(v) = \log(1+v)$, with $v \in \mathbf{V} \equiv (-1, \infty)$, will coincide with that of the dual feasibility conditions in Equation~\eqref{Eq3.17} (e.g., see \cite{r25} for some related discussions; see also \cite{r26} and \cite{r27} for general discussions on bilevel optimization problems). That is, the implied empirical probabilities are given by
\begin{align}
\hat{\pi}_\ell^{(r, k)} = \frac{1}{(N-m+1) \left(1 + (\hat{\lambda}^{(r, k)})^T \psi\bigl(\hat{z}_\ell^{(r, k)}, \hat{\beta}^{(r, k)} \bigr)\right)}, \quad \ell=1, \ldots, (n-m+1), \label{Eq3.16}
\end{align}
for $r =1$, $2$, \ldots, $R$, with $k \in \{1,\, 2, \ldots, n\}$. Moreover, noting that $\hat{\beta}^{(r, k)}$ and $\hat{\lambda}^{(r, k)}$ continuously depend on the parameters $\alpha^{(k, r)}$ and $\gamma^{(k, r)}$, then the empirical log-likelihood functions $\mathcal{L}^{(r, k)}\bigl(\alpha^{(r, k)}, \gamma^{(r, k)}\bigr)$ can be recovered as follows
\begin{align}
\mathcal{L}^{(r, k)}\bigl(\alpha^{(r, k)}, \gamma^{(r, k)}\bigr) &= - \log (N-m+1) \notag \\
& \quad \quad \quad - \sum\nolimits_{\ell=1}^{N-m+1} \log \left(1 + (\hat{\lambda}^{(r, k)})^T \psi\bigl(\hat{z}_\ell^{(r, k)}, \hat{\beta}^{(r, k)} \bigr)\bigr)\right), \label{Eq3.17}
\end{align}
for $r =1$, $2$, \ldots, $R$, with $k \in \{1,\, 2, \ldots, n\}$. Note that such an instantiated step will allows us to collect all gradient-related information for updating the estimated parameters $\bigl\{\phi^{(r, k)}\bigr\}_{r=1}^R = \bigl\{\alpha^{(r, k)}, \gamma^{(r, k)}\bigr\}_{r=1}^R$, with $k \in \{1,\,2, \ldots, n\}$, associated with the overall normalizing flows using first-order (stochastic) gradient methods. That is, starting from some initial values for $\alpha_0^{(r, k)}$ and $\gamma_0^{(r, k)}$, the updating schemes for the estimated parameters will proceed, for each iteration time step $t=1$, $2$, \ldots, $T-1$, as follows 
\begin{align}
\alpha_{t+1}^{(r, k)} = \alpha_{t}^{(r, k)} - \delta_t^{(r, k)} &\biggl(\nabla_{\alpha^{(r, k)}} \mathcal{L}^{(r, k)}\bigl(\alpha^{(r, k)}, \gamma^{(r, k)} \bigr) \notag\\
                                                               & \quad \quad+ \nu_{\alpha} \nabla_{\alpha^{(r, k)}}\mathcal{R}_{alpha}\bigl(\alpha^{(r, k)}, \gamma^{(r, k)}\bigr) \biggr) \biggl \vert_{\bigl(\alpha_t^{(r, k)}, \gamma_t^{(r, k)} \bigr)} \label{Eq3.18}
\end{align}
and
\begin{align}
\gamma_{t+1}^{(r, k)} = \gamma_{t}^{(r, k)} - \delta_t^{(r, k)} &\biggl(\nabla_{\gamma^{(r, k)}} \mathcal{L}^{(r, k)}\bigl(\alpha^{(r, k)}, \gamma^{(r, k)} \bigr) \notag\\
                                                               & \quad \quad+ \nu_{\gamma} \nabla_{\gamma^{(r, k)}}\mathcal{R}_{\gamma}\bigl(\alpha^{(r, k)}, \gamma^{(r, k)}\bigr) \biggr) \biggl \vert_{\bigl(\alpha_t^{(r, k)}, \gamma_t^{(r, k)} \bigr)} \label{Eq3.19}
\end{align}
where $\delta_t^{(r, k)}$ is the step size, $\mathcal{R}_{\alpha}\bigl(\alpha^{(r, k)}, \gamma^{(r, k)}\bigr)$ and $\mathcal{R}_{\gamma}\bigl(\alpha^{(r, k)}, \gamma^{(r, k)}\bigr)$ are the regularization terms, with $\nu_{\alpha} > 0$ and $\nu_{\gamma} > 0$. Note that, for any $k \in \{1,\,2, \ldots, n\}$, the above updating schemes allow us to estimate sequentially the parameters corresponding to the first $k$ sets of bijective transformations, in some sense of optimality, regardless of the remaining $(n-k)$ sets of bijective transformations. 

\subsection{A sequential decision-making paradigm with mixing distribution strategies for parameters aggregation scheme} \label{S3.3}
In this subsection, we specifically consider a decision-making paradigm with dynamic allocation scenarios that can be viewed as an extension mechanism for constructing flexible probability distributions as well as supporting smooth interpolation of parameters in high-dimensions. Here, at each iteration time step $t = 0, 1, 2, \ldots, T-1$, we decide on a set of mixing distribution strategies $\omega_t^{(r, k)}=\bigl\{\omega_t^{(1, k)}, \omega_t^{(2, k)}, \ldots, \omega_t^{(R, k)} \bigr\}$ over strategies, with $\omega_t^{(r, k)} > 0$, for $r =1$, $2$, \ldots, $R$, with $k \in \{1,2,\ldots, n\}$, and sum to one, i.e., $\sum\nolimits_{r=1}^R \omega_t^{(r, k)} = 1$, for all $k \in \{1,2,\ldots, n\}$, that will be used for dynamically apportioning weighting coefficients with respect to the current parameters estimates for the $k^{\rm th}$ set of bijective transforms. That is, with respect to the current parameters estimates $\bigl\{\phi^{(r, k)}\bigr\}_{r=1}^R = \bigl\{\alpha_{t}^{(r, k)}, \gamma_{t}^{(r, k)} \bigr\}_{r=1}^R$, with $k \in \{1,2,\ldots, n\}$,  we associate a set of risk factor measures $\rho_t^{(r, k)} \in [0,\,1]$, for  $r=1$, $2$, \ldots, $R$, with $k \in \{1,\, 2, \ldots, n\}$, based an exponential function, as follows
\begin{align}
\rho_t^{(r, k)} = 1 - \exp \left( - \zeta \bigl \Vert \Delta_t^{(r, k)} \bigr \Vert^{2} \right), \quad \zeta > 0, \label{Eq3.20}
\end{align}
for $t = 0, 1, 2, \ldots, T-1$, where the prediction errors $\Delta_t^{(r, k)}$, for $r =1$, $2$, \ldots, $R$, with $k \in \{1,2,\ldots, n\}$, are determined as follows
\begin{align}
\Delta_t^{(r, k)} =  \operatorname{T}_{\phi_t^{(r, k)}}^{-1} \bigl(z^{(r, k)}\bigr) - z^{(r, k-1)} \label{Eq3.21}
\end{align}
Note that the average (or mixture) loss, at each iteration time step $t = 0$, $1$,  $2$, \ldots, $T-1$, is given by  
\begin{align}
 \mathcal{J}_t^{(k)} = \sum\nolimits_{r=1}^R \omega_t^{(r, k)} \rho_t^{(r, k)}, \label{Eq3.22}
\end{align}
for $k \in \{1,2,\ldots, n\}$, with $k \in \{1,2,\ldots, n\}$. Then, our interest is to present an aggregation algorithm, that has a decision-theoretic interpretation with dynamic allocation scenarios, which also guarantees an upper bound for the sub-total mixture loss $\mathcal{J}_{\rm Total}^{(k)}$, i.e.,
\begin{align}
 \min \to \mathcal{J}_{\rm Total}^{(k)} &= \sum\nolimits_{t=0}^{T-1} \mathcal{J}_t^{(k)} \notag \\
                 &= \sum\nolimits_{t=0}^{T-1}\sum\nolimits_{r=1}^R \omega_t^{(r, k)} \rho_t^{(r, k)}, \label{Eq3.23}
\end{align}
for $k \in \{1,2,\ldots, n\}$. 

Moreover, such an aggregation algorithm also ensures the following additional properties:
\begin{enumerate} [(i)]
 \item $\rho_t^{(r, k)} $ tends to zero as $t \to \infty$ for $r=1$, $2$, \ldots, $R$ and for all $k \in \{1,2, \ldots, n \}$.
\item 
\begin{align*}
 \bigl\{\bar{\phi}_{T}^{(k)} \bigr\}_{k=1}^n &= \bigl\{ \bar{\alpha}_{T}^{(k)}, \bar{\gamma}_{T}^{(k)} \bigr\}_{k=1}^n\\
                                 &= \left\{\sum\nolimits_{r=1}^R \omega_T^{(r, k)} \alpha_{T}^{(r, k)}, \sum\nolimits_{r=1}^R \omega_T^{(r, k)} \gamma_{T}^{(r, k)} \right\}_{k=1}^n
\end{align*}
tends to the optimal estimated parameters $\phi^{\ast}$ for the overall aggregated normalizing flows.
\end{enumerate}
In order to accomplish the above two properties, we use a simple dynamic allocation strategy coupled with an iterative update scheme for computing the mixing distribution strategies $\omega_t^{(r, k)}$, i.e.,
\begin{align}
\omega_t^{(r, k)} = \frac{\mu_t^{(r, k)}}{\sum\nolimits_{r=1}^R \mu_t^{(r, k)}}, \quad r=1, 2, \ldots, R, \label{Eq3.24}
\end{align}
for $k \in \{1,2,\ldots, n\}$, with $k \in \{1,2,\ldots, n\}$; while, at each iteration time step $t = 0, 1, 2, \ldots, T-1$, the weighting coefficients are updated according to
\begin{align}
\mu_{t+1}^{(r, k)}= \mu_t^{(r, k)} \exp \bigl( \rho_t^{(r, k)} \log \varepsilon \bigr), \quad \varepsilon \in (0,1), \label{Eq3.25}
\end{align}
for $k \in \{1,2,\ldots, n\}$, with $k \in \{1,2,\ldots, n\}$. Moreover, we can assign arbitrary values for the initial weights $\mu_0^{(r, k)}$, for $r=1$, $2$, \ldots, $R$, with $k \in \{1,2, \ldots, n \}$,  but they must be nonnegative and sum to one.

Here, we remark that, for any $k \in \{1,\,2, \ldots, n\}$, the parameters corresponding to the first $k$ sets of bijective transformations are optimally determined and sequentially aggregated, regardless of the remaining $(n-k)$ sets of bijective transformations. Moreover, we state the following proposition that provides an upper bound for the overall total mixture loss, i.e., $\mathcal{J}_{\rm Total} = \sum\nolimits_{k=1}^n \mathcal{J}_{\rm Total}^{(k)}$ (cf. \cite[Proposition~1]{r28}).

\begin{proposition} \label{P1}
Let $\mathcal{J}_t^{(k)} = \sum\nolimits_{r=1}^R \omega_t^{(r, k)} \rho_t^{(r, k)}$, $t = 0$, $1$, $2$, \ldots, $T-1$, be a sequence of losses associated with a set of mixing distribution strategies $\omega_t^{(k)} = \bigl\{\omega_t^{(1, k)}, \omega_t^{(2, k)}, \ldots, \omega_t^{(R, k)}\bigr\}$ and risk factor measures $\rho_t^{(r, k)}$, for $r=1$, $2$,\ldots, $R$, with $k \in \{1,2, \ldots, n \}$. Then, the overall total mixture loss $\mathcal{J}_{\rm Total}$, i.e., 
\begin{align}
\mathcal{J}_{\rm Total} &= \sum\nolimits_{k=1}^n \mathcal{J}_{\rm Total}^{(k)} \notag\\
                                     &= \sum\nolimits_{k=1}^n \sum\nolimits_{t=0}^{T-1}\sum\nolimits_{r=1}^R  \omega_t^{(r, k)} \rho_t^{(r, k)}, \label{Eq3.26}
\end{align}
 satisfies the following upper bound condition
\begin{align}
\mathcal{J}_{\rm Total} \le -\frac{1}{1-\varepsilon} \sum\nolimits_{k=1}^n \log \left( \sum\nolimits_{r=1}^R \mu_{T}^{(r, k)} \right). \label{Eq3.27}
\end{align}
\end{proposition}

\begin{proof} Note that, for any $\varepsilon \in (0,1)$ and $\rho_t^{(r, k)} \in [0, 1]$, with$r=1,\, 2, \ldots, R$ and $k \in \{1,\, 2, \ldots, n\}$, we have the following inequality relation
\begin{align}
 \exp \bigl(\rho_t^{(r, k)} \log(\varepsilon)\bigr) \le 1 - (1- \varepsilon) \rho_t^{(r, k)}, \label{Eq3.28}
 \end{align}
due to the convexity argument for $\exp \bigl(\rho_t^{(r, k)} \log(\varepsilon)\bigr)\equiv \varepsilon^{\rho_t^{(r, k)}} \le 1 - (1 - \varepsilon)\rho_t^{(r, k)}$. Moreover, if we combine Equations~\eqref{Eq3.25} and \eqref{Eq3.28}, then we will have
\begin{align}
 \sum\nolimits_{r=1}^{R} \mu_{t+1}^{(r, k)}  &= \sum\nolimits_{r=1}^R \mu_{t}^{(r, k)} \exp \bigl(\rho_t^{(r, k)} \log(\varepsilon)\bigr) \notag\\
                                                                      &\le \sum\nolimits_{r=1}^R \mu_{t}^{(r, k)} \bigl(1 - (1- \varepsilon) \rho_t^{(r, k)} \bigr) \notag \\
                                                                      &= \sum\nolimits_{r=1}^R \mu_{t}^{(r, k)} -  (1- \varepsilon) \sum\nolimits_{r=1}^R \mu_{t}^{(r, k)} \rho_t^{(r, k)} \notag \\
                                                                     &= \left(\sum\nolimits_{r=1}^R  \mu_{t}^{(r, k)} \right) \left(1 - (1- \varepsilon) \sum\nolimits_{r=1}^R \omega_t^{(r, k)} \rho_t^{(r, k)}\right). \label{Eq3.29}
\end{align}
Furthermore, if we apply repeatedly for $t=0,1,2, \ldots, T-1$ to the above equation, then we have
{\small 
\begin{align}
\sum\nolimits_{r=1}^R \mu_t^{(r, k)} &\le \prod\nolimits_{t=0}^{T-1}\left(1 - (1- \varepsilon) \sum\nolimits_{r=1}^R \omega_t^{(r, k)} \rho_t^{(r, k)} \right) \notag \\
                                                                &\le \exp \left( - (1- \varepsilon) \sum\nolimits_{t=0}^{T-1}\sum\nolimits_{r=1}^R \omega_t^{(r, k)} \rho_t^{(r, k)} \right). \label{Eq3.30}
\end{align}}

Note that, due to the inequality relation $1 + x \le \exp(x)$ for all $x$, the right-hand side of the above equation satisfies the following inequality
\begin{align}
 1 &- (1- \varepsilon) \sum\nolimits_{t=0}^{T-1}\sum\nolimits_{r=1}^R \omega_t^{(r, k)} \rho_t^{(r, k)}  \notag\\
 & ~ \le \exp \left( - (1- \varepsilon) \sum\nolimits_{t=0}^{T-1}\sum\nolimits_{r=1}^R \omega_t^{(r, k)} \rho_t^{(r, k)} \right), \label{Eq3.31}
\end{align}
that further gives us the following result
\begin{align}
 \sum\nolimits_{r=1}^R  \mu_T^{(r, k)} \le \exp \left( - (1- \varepsilon) \sum\nolimits_{t=0}^{T-1}\sum\nolimits_{r=1}^R \omega_t^{(r, k)} \rho_t^{(r, k)} \right), \label{Eq3.32}
\end{align}
for each $k \in \{1,2, \ldots, n \}$. Moreover, if we take the natural logarithm on both sides of the above equation and sum over all $k$ (ranging from $k=1$ to $n$), i.e.,
\begin{align}
\sum\nolimits_{k=1}^n \left( \log \sum\nolimits_{r=1}^R  \mu_T^{(r, k)}\right) \le -(1- \varepsilon) \sum\nolimits_{k=1}^n \sum\nolimits_{t=0}^{T-1}\sum\nolimits_{r=1}^R \omega_t^{(r, k)} \rho_t^{(r, k)}. \label{Eq3.33}
\end{align}
Then, the statement in the proposition follows immediately. 
\end{proof}

Here, we remark that the dynamic allocation strategies together with the iterative updating schemes (i.e., Equations~\eqref{Eq3.24} and \eqref{Eq3.25}) provide conditions under which such an algorithm determines a set of mixing distribution strategies used for aggregating the current estimated parameters for the overall normalizing flows that ultimately leading to an optimal parameter estimate, i.e., as consensus solutions, supporting smooth interpolation of parameters in high-dimensions and further facilitating an extension mechanism for constructing flexible probability distributions.

Finally, it is worth remarking that the risk factor measures $\rho_t^{(r, k)}$, for $r=1$, $2$,\ldots, $R$, with $k \in \{1,2, \ldots, n \}$, tend to zero as $t \to \infty$. Moreover, based on the averaged (aggregated) values of the current estimated parameters, the updating schemes in Equations~\eqref{Eq3.18}  and \eqref{Eq3.19} can be modified, for $t=1$, $2$, \ldots, $T-1$, as follows 
\begin{align}
\alpha_{t+1}^{(r, k)} = \bar{\alpha}_{t}^{(k)} - \delta_t^{(r, k)} &\biggl(\nabla_{\alpha^{(r, k)}} \mathcal{L}^{(r, k)}\bigl(\alpha^{(r, k)}, \gamma^{(r, k)} \bigr) \notag\\
                                                               & \quad \quad+ \nu_{\alpha} \nabla_{\alpha^{(r, k)}}\mathcal{R}\bigl(\alpha^{(r, k)}, \gamma^{(r, k)}\bigr) \biggr) \biggl \vert_{\bigl(\bar{\alpha}_t^{(k)}, \bar{\gamma}_t^{(k)} \bigr)} \label{Eq3.34}
\end{align}
and
\begin{align}
\gamma_{t+1}^{(r, k)} = \bar{\gamma}_{t}^{(r, k)} - \delta_t^{(r, k)} &\biggl(\nabla_{\gamma^{(r, k)}} \mathcal{L}^{(r, k)}\bigl(\alpha^{(r, k)}, \gamma^{(r, k)} \bigr) \notag\\
                                                               & \quad \quad+ \nu_{\gamma} \nabla_{\gamma^{(r, k)}}\mathcal{R}\bigl(\alpha^{(r, k)}, \gamma^{(r, k)}\bigr) \biggr) \biggl \vert_{\bigl(\bar{\alpha}_t^{(k)}, \bar{\gamma}_t^{(k)} \bigr)}, \label{Eq3.35}
\end{align}
where the current aggregated estimated parameters $\bar{\alpha}_t^{(k)}$ and $\bar{\gamma}_t^{(k)}$, for $k=1$, $2$, \ldots, $n$, are, respectively, given by
\begin{align*}
\bar{\alpha}_t^{(k)} =  \sum\nolimits_{r=1}^R \omega_t^{(r, k)} \alpha_{t}^{(r, k)} \quad \text{and} \quad \bar{\gamma}_t^{(k)} =  \sum\nolimits_{r=1}^R \omega_t^{(r, k)}\gamma_{t}^{(r, k)}. 
\end{align*}
Note that estimated parameters $ \bigl\{\bar{\phi}_{T}^{(k)} \bigr\}_{k=1}^n = \bigl\{\bar{\alpha}_{T}^{(k)}, \bar{\gamma}_{T}^{(k)}\bigr\}_{k=1}^n$ tend to the optimal estimated parameters $\phi^{\ast}$ for the overall aggregated normalizing flows.

\subsection{Further remarks: Inference and quantifying model performances} \label{S3.4}
Note that the task of inference may notably vary across different scientific fields by what is being inferred. However, they all share the same principle, i.e., given real observation data $y_{\rm obs}$ from a certain physical system or experimental process, then the goal is to make inferring about $\theta$ in the parametric setting, with the support of $\Theta \subset \mathbb{R}^{d_{\theta}}$, or extending the inference tasks to include a probabilistic notion of uncertainty and quantifying model performances. In all cases, the critical component for the inference task is the likelihood function which is typically not tractable, since it is implicitly defined by the forward computational simulation that requires all possible execution traces of the simulator. In what follows, we provide a few remarks with regard to inference and quantifying model performances.
 
\subsubsection*{Inference} In what follows, having estimated the optimal parameters $\phi^{\ast} = \bigl\{ \bar{\alpha}_{T}^{(k)},\bar{\gamma}_{T}^{(k)} \bigr\}_{k=1}^n$ for the aggregated overall normalizing flows, where the later serve as a proxy for the true likelihood function, i.e.,
\begin{align}
\mathcal{L}_{\rm Approx} \bigl(y_{\rm obs} \vert \theta, \phi^{\ast} \bigr) \approx p\bigl(y_{\rm obs} \vert \theta\bigr). \label{Eq3.36}
\end{align}
Then, we will proceed with the problem of inferring the model parameters $\theta$ consistent with respect to real experimental observation datasets $y_{\rm obs} = \bigl\{y_{\rm obs}^{(j)} \bigr\}_{j=1}^N$, where each $y_{\rm obs}^{(j)}$, for $j=1$, $2$, \ldots, $N$, represents an experimental measurement. Note that, from Bayes' theorem, the probability distribution for $\theta$ given $y_{\rm obs}$, i.e., the posterior distribution $p\bigl(\theta \vert y_{\rm obs}\bigr)$, is given by
\begin{align}
p\bigl(\theta \vert y_{\rm obs}\bigr) &= \frac{p\bigl( y_{\rm obs} \vert \theta \bigr) \pi\bigl(\theta\bigr)}{p\bigl(y_{\rm obs} \bigr)} \notag \\
                                                             &\propto p\bigl(y_{\rm obs} \vert \theta\bigr) p\bigl(\theta\bigr),  \label{Eq3.37}
\end{align}
where $\pi\bigl(\theta\bigr)$ is the prior distribution over the parameters; while $p\bigl(y_{\rm obs} \bigr)$ is the evidence, which is independent of the model parameters $\theta$, and thus allowing us to rewrite the posterior distribution up to a proportionality constant as above. Furthermore, if we assume that the observation datasets are conditionally independent, given the model parameters $\theta$, then the total likelihood will be factorized as follows
\begin{align}
p\bigl(y_{\rm obs} \vert \theta\bigr) = \prod\nolimits_{j=1}^N p\bigl(y_{\rm obs}^{(j)} \vert \theta\bigr).  \label{Eq3.38}
\end{align}
Moreover, we can replace each of the likelihood terms with that of the output from the overall aggregated normalizing flows, with the optimal estimated parameters $\phi^{\ast}$, and conditioned on the model parameters $\theta$. That is, setting $p\bigl(y_{\rm obs} \vert \theta\bigr) \approx \mathcal{L}_{\rm Approx} \bigl(y_{\rm obs} \vert \theta, \phi^{\ast} \bigr)$. Note that, if we substitute this into the factorized likelihood function. Then, the posterior distribution up to proportionality is given by
\begin{align}
p\bigl(\theta \vert y_{\rm obs}\bigr) \propto \left (\prod\nolimits_{j=1}^N \mathcal{L}_{\rm Approx} \bigl(y_{\rm obs}^{(j)}\vert \theta, \phi^{\ast} \bigr) \right)\pi\bigl(\theta\bigr).  \label{Eq3.39}
\end{align}
Then, our task will include to use appropriate sampling algorithms such as the Metropolis-Hastings MCMC sampling or the adaptive multiple importance sampling algorithms, that increase our knowledge most or allow us to explore the approximate likelihood function $\mathcal{L}_{\rm Approx} \bigl(y_{\rm obs} \vert \theta, \phi^{\ast} \bigr)$ with respect to the model parameters $\theta$ consistent with the real observation datasets $y^{\rm obs}$ (e.g., see \cite{r29}-\cite{r33} for additional discussions). Here, we remark that the choice of appropriate algorithms as well as a set of suitable proposal distributions are consequential for the convergence of such algorithms.

\subsubsection*{Quantifying model performances} For evaluating the model performances, we can use typical statistical distance measures, such as Wasserstein, Cram\'{e}r-von Mises, and Kolmogorov-Smirnov distances, that will allow us to quantify the differences between probability distributions during model parameters estimation step. Note that, if we use the Cram\'{e}r-von Mises distance measure as a sub-problem optimization criterion during model parameters estimation. Then, due to its balanced sensitivity to both global distribution shifts and local deformations, such a distance measure will provide us the most robust translation for quantifying the model performance, i.e., a balanced effect when evaluating the $L_2$-norm distance measures across multiple calibration sources with diverse characteristics.

\end{document}